\documentclass[twocolumn]{autart}    

\usepackage{amsmath,graphicx,amsfonts,amssymb,epsfig,subfig,mathrsfs,mathtools,newproof}
\usepackage[table,xcdraw]{xcolor}
\usepackage{xfrac,nicefrac}
\usepackage{stfloats}
\usepackage{arydshln}
\usepackage{blkarray}

\usepackage{color}
\usepackage{lipsum}
\usepackage{algorithm}
\usepackage{algpseudocode}
\usepackage{tikz}

{}
\newtheorem{corollary}{Corollary}{}
\newtheorem{remark}{Remark}{}
\newtheorem{assumption}{Assumption}{}
\newtheorem{lemma}{Lemma}{}
\newtheorem{proposition}{Proposition}{}
\newtheorem{theorem}{Theorem}{}

\begin{document}

\begin{frontmatter}

\title{Asynchronous Distributed Averaging: A Switched System Framework for Average Error Analysis\thanksref{footnoteinfo}} 

\thanks[footnoteinfo]{This paper was not presented at any IFAC 
meeting. Corresponding author K.~Lee. Tel. +1-575-835-5554. 
Fax +1-575-835-5209.}

\vspace{-0.2in}
\author[Lee]{Kooktae Lee}\ead{kooktae.lee@nmt.edu}    

\address[Lee]{Department of Mechanical Engineering, New Mexico Institute of Mining and Technology, Socorro, NM 87801, USA}  

\begin{keyword}                           
Distributed averaging; Asynchronous communication; Multi-agent average consensus; Error analysis.               
\end{keyword}                             

\begin{abstract}                          
This paper investigates an expected average error for distributed averaging problems under asynchronous updates. The asynchronism in this context implies no existence of a global clock as well as random characteristics in communication uncertainty such as communication delays and packet drops. Although some previous works contributed to the design of average consensus protocols to guarantee the convergence to an exact average, these methods may increase computational burdens due to extra works. Sometimes it is thus beneficial to make each agent exchange information asynchronously without modifying the algorithm, which causes randomness in the average value as a trade-off. In this study, an expected average error is analyzed based on the switched system framework, to estimate an upper bound of the asynchronous average compared to the exact one in the expectation sense. Numerical examples are provided to validate the proposed results.
\end{abstract}
\end{frontmatter}

\section{Introduction}
Consider an undirected graph $\mathcal{G}=\{\mathcal{V},\mathcal{E}\}$, where $\mathcal{V}=\{\textit{v}_1,\textit{v}_2,\hdots,\textit{v}_n\}$ is a set of $n$ numbers of nodes, $\mathcal{E}\subseteq \mathcal{V}\times \mathcal{V}$ is a set of edges. The symbol $\mathcal{N}_i=\{j\vert\{i,j\}\in\mathcal{E}\}$ represents the set of neighbors of node $i$.
Given an initial value $x_i(0)$ for each node $i$, a distributed averaging problem is to seek an average of $x(0)=[x_1(0), x_2(0),\ldots,x_n(0)]^{T}$, i.e. $\overline{x}:=\dfrac{1}{n}\underline{1}^{T}x(0)$, through information exchanges with connected nodes.

To reach the exact average in distributed averaging, it is known \cite{lee2019effect} that the synchronization is necessary. This is quite restrictive in reality because of some implementation issues like the existence of a global clock as well as communication uncertainty (e.g., communication delays and packet drops). Although an asynchronous algorithm, a counterpart of the synchronous one, naturally takes into account the above issues and hence more practical, it has been reported in \cite{lee2019effect}, \cite{xiao2008asynchronous}, \cite{fang2005information}, \cite{acciani2018achieving} that asynchronous information exchanges may lead to inaccuracy in solutions. Therefore, robust communication protocols have been studied (see \cite{mehyar2007asynchronous}, \cite{chen2010corrective}, \cite{hadjicostis2016robust}, \cite{du2020accurate} to list a few) to correct biased averages.
These protocols serve as a compensator to induce an exact average while allowing asynchronous communications between connected nodes. These approaches, however, may increase computational burdens due to auxiliary algorithms in each computing node, possibly delaying the time for averaging. Sometimes it is thus more advantageous to enable each agent to exchange information asynchronously without modifying the algorithm, which necessitates rigorous analysis on the accuracy of resultant asynchronous averages.

In this work, expected average error analysis is performed to estimate an upper bound of the asynchronous average compared to the exact one in the expectation sense. The switched system framework \cite{lee2015stability} is adopted as a tool to analyze the expected average error. The proposed method does not require a node to know the entire information about the interaction topology. Rather, knowing diagonal elements in the interaction topology matrix is enough for the error analysis. Moreover, only the largest absolute initial value, but not all of them, needs to be exposed for the upper bound calculation and thus, it preserves the privacy of each node. To validate the proposed results, numerical examples and results are provided.


\section{Preliminaries}

\noindent\textbf{Notation:} 
The set of real and natural numbers are denoted by $\mathbb{R}$ and $\mathbb{N}$. Moreover, $\mathbb{N}_0 := \mathbb{N}\cup\{0\}$. The superscript $^{T}$ is given to represent a transpose operator. The symbols $\text{I}^{n\times n}$ and $\text{0}^{n\times n}$ stand for an ${n\times n}$ identity and zero matrix, respectively, whereas the dimension is omitted in an apparent case. Further, $\underline{1}$ and $\underline{0}$, respectively, represent a column vector with all components being one and zero. The Euclidean and infinity norms are expressed as $\lVert \cdot \rVert$ and $\lVert \cdot \rVert_{\infty}$, respectively. The symbol $\otimes$ denotes the Kronecker product and $\mathbb{E}[\cdot]$ stands for the expectation with a given probability. The operator $\text{diag}(\cdot)$ returns a column vector composed of diagonal elements for a given square matrix. For any $X\in\mathbb{R}^{m\times n}$, $X^{\otimes r}$ is defined by $X^{\otimes r}:=\underbrace{X\otimes X \otimes \cdots \otimes X}_{r \text{ times}}$.

Under asynchronous communications, the state update model in each node for distributed averaging can be written as follows.
\begin{align}
x_{i}(k_i+1) = a_{ii}x_{i}(k_i) + \sum_{j\in\mathcal{N}_i}a_{ij}x_{j}(k_j^*),\label{eqn:each state async}
\end{align}
where $A := [a_{ij}]\in\mathbb{R}^{n\times n}$ is a \textit{doubly stochastic} matrix and $k_i\in\mathbb{N}_0$ is a discrete-time index that can be different from each agent, meaning a global clock is unnecessary. Given a definition of a set $\mathbb{K}_q:=\{k,k-1,\ldots,k-q+1\}$, where $q\in\mathbb{N}$, the random variable $k_j^*\in\mathbb{K}_q$ represents asynchronism.

In the asynchronous setup, the node $i$ does not wait for the values being received from the neighboring nodes. Rather, it executes the update law \eqref{eqn:each state async} with the most recent value $x_j(k_j^{*})$ saved in the buffer memory.
If new values are received from neighboring nodes, each $x_j(k_j^*)$ is updated accordingly.
The following assumptions are proposed for the asynchronous update model \eqref{eqn:each state async}.

\begin{assumption}\label{assump:eventual_comm.}
For any node $i$ and any given $k_j^{*}\in\mathbb{K}_q$, where $j\in\mathcal{N}_i$, there exists time $k_i\geq k_j^{*}$ such that $x_j$ is updated by communications with the node $j$ at time $k_i$. 
\end{assumption}
\begin{assumption}\label{assump:iter_no_more_than_once}
If $\Delta t_i$ is denoted by actual elapsed time to update $x_i$ using \eqref{eqn:each state async}, then for any node $i$ this update cannot occur more than once over the time period $\max_i \Delta t_i$.

\end{assumption}
Assumption \ref{assump:eventual_comm.} is identical to the eventual update assumption in \cite{mehyar2007asynchronous} and Assumption \ref{assump:iter_no_more_than_once} assures that the state update across all nodes cannot happen more than once for a given time period, $\max_i \Delta t_i$.
The average error analysis will be carried out under these assumptions.


\section{Average Error Analysis}
Asynchronous updates in multi-agent consensus problems may lead to inaccuracy in consensus values \cite{lee2019effect}. This implies that the asynchronous distributed average consensus may end up with an inexact average. As proposed in \cite{mehyar2007asynchronous}, \cite{chen2010corrective}, \cite{hadjicostis2016robust}, \cite{du2020accurate}, this issue can be obviated by augmenting auxiliary algorithms that require extra computing, possibly delaying the time to obtain the average. Sometimes it is beneficial in terms of computational speed and easiness of implementation to proceed with distributed averaging asynchronously, which necessitates rigorous analysis on how far the asynchronous average is deviated from the exact average. 

For this purpose, the switched system widely adopted dynamical system framework to describe uncertainty in a communication network is introduced as follows.
\subsection{Switched System}
Given a definition ${y}(k) := [x(k)^{T}, x(k-1)^{T}, \hdots, x(k-q+1)^{T}]^{T}$ with $k\in\mathbb{N}_0$ being a discrete-time index and $x(k):= [x_1(k), x_2(k),\hdots, x_n(k)]^{T}$, the asynchronous dynamics \eqref{eqn:each state async} evolves under Assumptions \ref{assump:eventual_comm.} $-$ \ref{assump:iter_no_more_than_once} by
\begin{align}
{y}(k+1) &= W_{\sigma_k}{y}(k),\label{eqn:switched system}
\end{align}
where $W_{\sigma_k}\in\{W_j\}_{j=1}^{\eta}$ denotes a modal matrix with a switching mode $\sigma_{k}$ at time $k$
and the total number of switching modes $\eta$.

The modal matrices $W_j$, $j=1,2,\hdots,\eta$, are obtained by taking into account every possible scenarios for asynchronous communications.
Based on the fact that asynchrony only takes place while communicating with neighboring agents, the diagonal elements $\{a_{ii}\}$ in $A$ are always stationary in all $W_j$, whereas off-diagonal elements move accordingly in $W_j$. The random characteristic of asynchrony is delineated by a switching rule that controls the switching process $\{\sigma_k\}_{k=0}^{\infty}$. The switched system with a stochastic switching process is particularly referred to as the stochastic switched system or stochastic jump linear system \cite{lee2015performance}. 
The most general form of the modal matrix $W_{j}$ has the following structure:

\vspace{-0.3in}
\small
\begin{align}
W_{j} \in\mathbb{R}^{nq\times nq}&=\begin{bmatrix}
W_{11}(k) & W_{12}(k) & W_{13}(k) & \cdots & W_{1q}(k)\\
\text{{I}}^{n\times n} & \text{{0}}^{n\times n}& \text{{0}}^{n\times n} & \cdots &\text{{0}}^{n\times n}\\ 
\text{{0}}^{n\times n} & \text{{I}}^{n\times n}& \text{{0}}^{n\times n} & \cdots &\text{{0}}^{n\times n}  \\
\vdots & \text{{0}}^{n\times n} &\ddots & \ddots& \vdots\\
\text{{0}}^{n\times n} & \text{{0}}^{n\times n} & \ddots & \text{{I}}^{n\times n} & \text{{0}}^{n\times n}
\end{bmatrix},&\hfill\label{eqn:W_j general}
\end{align}
\normalsize
where $W_{1l}(k)\in\mathbb{R}^{n\times n}$, $l=1,2,\ldots, q$ in $W_j$ is a block matrix such that $\sum_{l=1}^{q}W_{1l}(k) = A$ and the diagonal elements in $W_{11}(k)$ is the same with that in $A$ for all $\forall k$.
Notice that all $W_{j}$ matrices are row stochastic but not column stochastic as each row sum is unity with all nonnegative elements.

As is well known (\cite{lee2015stability}, \cite{lee2015performance}), the switched system framework can be used to formulate the inherent dynamics of the asynchronous update model. 
The following assumptions are given to facilitate the average error analysis under asynchronous updates.

\begin{assumption}\label{assump:i.i.d.}
In \eqref{eqn:each state async}, $k_j^{*}\in\mathbb{K}_q$ is a random variable governed by an i.i.d. (independent and identically distributed) probability $\pi=[\pi_1,\ldots, \pi_q]$, where $\pi_j$ denotes a discrete probability for $j-1$ step delays from neighboring agents.  
\end{assumption}
\begin{assumption}\label{assump:known_F_at_leat_one}
At least, one of the nodes has the information about the diagonal elements $a_{ii}$, $i=1,2,\hdots,n$, in $A$.
\end{assumption}

Assumption \ref{assump:i.i.d.} is given to concertize the characteristics of asynchrony and the i.i.d. probability can be statistically obtained through the collection of data while carrying out iterations. Also, the i.i.d. probability is assumed to be identical across all agents.
Assumption \ref{assump:known_F_at_leat_one} guarantees that at least one node has the information about the diagonal elements in the matrix $A$ and hence, this node can perform the average error analysis. In what follows, we develop the expected average error analysis based on above assumptions.

\begin{lemma}[K. Lee \cite{lee2019effect}]\label{lemma:switching_prob.}
Under Assumption \ref{assump:i.i.d.}, a switching probability governing the switching process $\{\sigma_k\}$ for \eqref{eqn:switched system} is given by $\nu =  [\nu_1,\ldots,\nu_{\eta}] = \pi^{\otimes n(n-1)}$, which is $n(n-1)$ times of the Kronecker product of $\pi\in\mathbb{R}^{1\times q}$, where each $\nu_i$ represents the modal probability associated with the modal matrix $W_j$ and $\eta=q^{n(n-1)}$ is the total number of the switching modes.
\end{lemma}

As both $\pi$ and $\nu$ are discrete probability distributions, it follows $\sum_{j=1}^{q}\pi_j=1$ and $\sum_{i=1}^{\eta}\nu_i = 1$.

\begin{proposition}[K. Lee \cite{lee2019effect}]\label{prop:Wtilde_reduced_model}
Consider the switched system \eqref{eqn:switched system} with a switching probability $\nu$ as in Lemma \ref{lemma:switching_prob.}.
For a row vector $s_j\in\mathbb{R}^{1\times q}$, defined by all elements with zero except the $j^{th}$ element being one, i.e., $s_j := [0, \ldots, 0, 1, 0 ,\ldots , 0]$, and the following matrices

\vspace{-0.3in}
\small
\begin{align}
W^{\text{diag}} &:= \begin{bmatrix}
A^{\text{diag}} & \text{{0}}^{~n\times (n-1)q} \\
\\
\text{{I}}^{~(n-1)q\times (n-1)q} & \text{{0}}^{~(n-1)q\times n}
\end{bmatrix}, \\
W_j^{\text{off}} &:= 
s_j\otimes \begin{bmatrix}
A^{off}\\
\text{{0}}^{~(n-1)q\times n}
\end{bmatrix},\label{eqn:W_j^off}
\end{align}
\normalsize
the expectation of $W_{\sigma_k}$, $\overline{W}:=\mathbb{E}[W_{\sigma_k}]$, at any time $k$ is calculated by 
\begin{align}
\overline{W} = \sum_{i=1}^{\eta=q^{n(n-1)}}\nu_i W_i = W^{diag} + \sum_{j=1}^{q}\pi_j W_j^{off},\label{eqn: Wmean equivalent form}
\end{align}
where $A^{\text{diag}}$ and $A^{off}$, respectively, are the matrices composed of diagonal and off-diagonal elements in $A$.
\end{proposition}

\begin{remark}[Scalability issue]\label{remark: scalability}
In Lemma \ref{lemma:switching_prob.} and the first equality of \eqref{eqn: Wmean equivalent form}, constructing $\overline{W}$ requires the information about all $W_j$, $j=1,2,\ldots,\eta$. This computation causes the notorious \textbf{scalability issue} as indicated in \cite{lee2015stability}, \cite{lee2015performance} due to the scale of $\eta=q^{n(n-1)}$. 
Proposition \ref{prop:Wtilde_reduced_model} assures that the computational complexity of $\overline{W}$ drastically reduces 
as the size of summation terms in the second equality of \eqref{eqn: Wmean equivalent form} only grows \textit{linearly} with respect to $q$.
As a result, the scalability issue can be avoided by Proposition \ref{prop:Wtilde_reduced_model}.
\end{remark}

\begin{theorem}\label{thm: 1}
Consider the asynchronous distributed averaging problem under Assumptions $1-4$. 
For the exact and asynchronous average denoted by $\bar{x}$ and $x^{\star}$, respectively, the expected average error, $\lvert \mathbb{E}\left[\bar{x} - x^{\star}\right]\rvert$, is upper bounded by
\begin{align}
\lvert \mathbb{E}\left[\bar{x} - x^{\star}\right]\rvert \leq 
\dfrac{c\sqrt{n}}{d} \lVert \text{diag}(A - \overline{a}_{ii}I) \rVert \cdot \lVert x(0)\rVert_{\infty},\label{eqn:thm 1}
\end{align}
where $c:=\sum_{j=2}^{q}(j-1)\pi_j$, $d:=n\left((1+c)-c\overline{a}_{ii}\right)$, and $\overline{a}_{ii}:=\dfrac{1}{n}\sum_{i=1}^{n}a_{ii}$.
\end{theorem}

\begin{proof}
Let $w^{\star}:=[w_1^{\star},w_2^{\star},w_3^{\star},\ldots,w_q^{\star}]\in\mathbb{R}^{1\times nq}$ be the \textit{left} eigenvector of $\overline{W}$ with an eigenvalue one, i.e., $w^{\star}\overline{W}=w^{\star}$, where $w_j^{\star}\in\mathbb{R}^{1\times n}$. It is worth noting that the stationary form of $\overline{W}$ is given by $\overline{W}^{\star}= \underline{1}\otimes w^{\star}$. Then, from the intrinsic structure of $\overline{W}$ similar to \eqref{eqn:W_j general}, the computation of $w^{\star}\overline{W}=w^{\star}$ yields
\begin{equation}
\begin{aligned}
&w_1^{\star}\overline{W}_{11} + w_2^{\star} = w_1^{\star}\\
&w_1^{\star}\overline{W}_{12} + w_3^{\star} = w_2^{\star}\\
&\qquad\qquad\vdots\\
&w_1^{\star}\overline{W}_{1(q-1)} + w_{q}^{\star} = w_{q-1}^{\star}\\
&w_1^{\star}\overline{W}_{1q} = w_{q}^{\star}
\end{aligned}\label{eqn: w_jWbar=w_j}
\end{equation}
resulting in
\vspace{-0.3in}
\small
\begin{align*}
w^{\star}&=[w_1^{\star},w_1^{\star}\left(I-\overline{W}_{11}\right),w_1^{\star}\left(I-\overline{W}_{11}-\overline{W}_{12}\right),\ldots,\\
&\qquad\qquad w_1^{\star}\left(I-\overline{W}_{11}-\overline{W}_{12}-\cdots -\overline{W}_{1(q-1)}\right)].
\end{align*}
\normalsize
Given a definition $(w^{\star})':=\sum_{j=1}^{q}w_j^{\star}$, it follows
\vspace{-0.3in}
\small
\begin{align*}
(w^{\star})' &= qw_1^{\star} - (q-1)w_1^{\star}\overline{W}_{11} - \cdots - w_1^{\star}\overline{W}_{1(q-1)}\\
&= qw_1^{\star} - \sum_{j=1}^{q-1}(q-j)w_1^{\star}\overline{W}_{1j}\\
&= w_1^{\star}(qI - \sum_{j=1}^{q-1}(q-j)\overline{W}_{1j}).
\end{align*}
\normalsize
From the result in Proposition \ref{prop:Wtilde_reduced_model}, $\overline{W}_{11} = A^{diag} + \pi_1A^{off}$ and $\overline{W}_{1j} = \pi_jA^{off}$ for $j=2,3,\ldots, q$. Then, the above equation is equivalently written by
\vspace{-0.3in}
\small
\begin{align*}
&(w^{\star})' = w_1^{\star}\Big(qI - (q-1)A^{diag}-\sum_{j=1}^{q-1}(q-j)\pi_jA^{off}\Big)\\
&= w_1^{\star}\Big(qI - (q-1)A^{diag}\\
&\qquad\qquad\qquad\qquad-((q-1)\pi_1+\sum_{j=2}^{q-1}(q-j)\pi_j)A^{off}\Big)\\
&= w_1^{\star}\Big(qI - (q-1)A^{diag}-((q-1)-\sum_{j=2}^{q}(j-1)\pi_j)A^{off}\Big)\\
&= w_1^{\star}\Big(qI - (q-1)A + \sum_{j=2}^{q}(j-1)\pi_jA^{off} \Big),
\end{align*}
\normalsize
where the third equality holds by $\pi_1 = 1-\sum_{j=2}^{q}\pi_j$ and the following information $A=A^{diag}+A^{off}$ is used in the last equality.

Adding all equations in \eqref{eqn: w_jWbar=w_j} with the fact that $\sum_{j=1}^{q}\overline{W}_{1j}$ $= A$ yields
$w_1^{\star}A = w_1^{\star}$.
Since the matrix $A$ is doubly-stochastic, we have
$w_1^{\star}=\underline{1}^{T}$, leading to
\vspace*{-0.3in}
\small
\begin{align*}
(w^{\star})' &= \left(q - (q-1)\right)\underline{1}^{T} + \sum_{j=2}^{q}(j-1)\pi_j\underline{1}^{T}A^{off}\\
&= [1+c(1-a_{11}),1+c(1-a_{22}), \ldots, 1+c(1-a_{nn})],
\end{align*}
\normalsize
where $c$ is given in \eqref{eqn:thm 1}.

The normalized form of $(w^{\star})'$ is then obtained by
\begin{align*}
(w^{\star})'_{\text{normal}} = \dfrac{1}{n(1+c)-c\sum_{i=1}^{n}a_{ii}}(w^{\star})'.
\end{align*}

Finally, from the fact that $\bar{x}=\dfrac{1}{n}\underline{1}^{T} x(0)$ and $\mathbb{E}[x^{\star}]=(w^{\star})'x(0)$, the expected average error is calculated by
\begin{align*}
&\lvert \mathbb{E}\left[\bar{x} - x^{\star}\right]\rvert \\
&= \left\lvert \left(\dfrac{1}{n}\underline{1}^{T} - \dfrac{1}{n(1+c)-c\sum_{i=1}^{n}a_{ii}}(w^{\star})'\right)x(0)\right\rvert\\
&\leq \left\lVert \dfrac{1}{n}\underline{1}^{T} - \dfrac{1}{n(1+c)-c\sum_{i=1}^{n}a_{ii}}(w^{\star})' \right\rVert \cdot \lVert x(0)\rVert\\
&= \left\lVert \dfrac{c}{d}[a_{11} - \overline{a}_{ii}, \ldots, a_{nn}-\overline{a}_{ii}] \right\rVert \cdot \lVert x(0)\rVert\\
&\leq\dfrac{c}{d} \lVert \text{diag}(A - \overline{a}_{ii}I) \rVert \cdot \sqrt{n}\lVert x(0)\rVert_{\infty},
\end{align*}
where $d$ and $\overline{a}_{ii}$ are defined in \eqref{eqn:thm 1} and the last inequality holds by the vector norm property.
\end{proof}

\begin{remark}[Privacy issue]
In Theorem \ref{thm: 1}, the upper bound for the expected average error is obtained using $\sqrt{n}\lVert x(0)\rVert_{\infty}$, instead of $\lVert x(0) \rVert$. Although this leads to conservatism, computation of $\lVert x(0) \rVert$ requires the information about all initial values of each node (opinions), causing an infringement of privacy issue in some applications (e.g., social networks). On the other hand, the infinity norm only requires the largest absolute value in $x(0)$, which is possibly informed to the node knowing $\{a_{ii}\}$ (under Assumption 4) through message passing. In this way, the privacy intrusion issue can be averted as it is not required to know opinions of all nodes.
\end{remark}

\begin{corollary}\label{cor: 1}
For the asynchronous distributed averaging problem \eqref{eqn:each state async} under Assumptions $1-3$, there is no average error irrespective of asynchronism in the expectation sense, i.e., $\lvert \mathbb{E}\left[\bar{x} - x^{\star}\right]\rvert=0$, if $a_{ii}$ in $A$ is identical across all nodes $i=1,2,\ldots,n$.
\end{corollary}
\begin{proof}
In the case that $a_{ii} = a_{jj}$ $\forall i,j$, we have $a_{ii} = \overline{a}_{ii}$, $\forall i$, yielding $\text{diag}(A - \overline{a}_{ii}I) = \underline{0}$ in \eqref{eqn:thm 1} and hence, $\lvert \mathbb{E}\left[\bar{x} - x^{\star}\right]\rvert=0$.
\end{proof}

\begin{remark}[Independent probability]
The result in Corollary \ref{cor: 1} holds even if $\pi$ is not i.i.d. It is enough for this probability being independent for each event.
\end{remark}

\section{Numerical Example}
\begin{figure*}[t]
\centering
\subfloat[]{\includegraphics[scale=0.38]{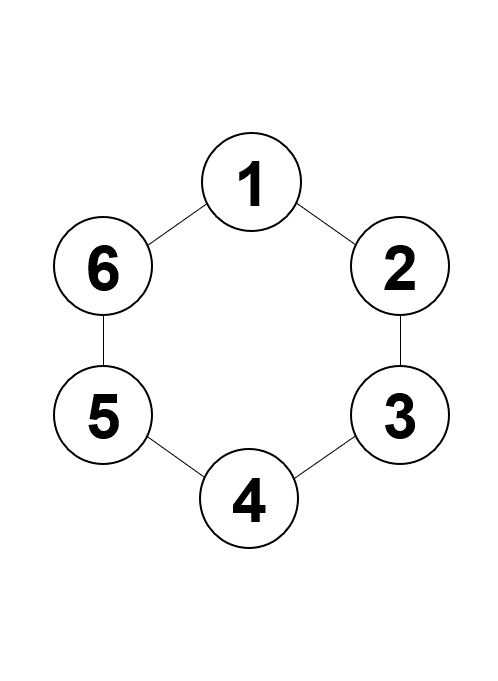}}\qquad
\subfloat[]{\includegraphics[scale=0.2]{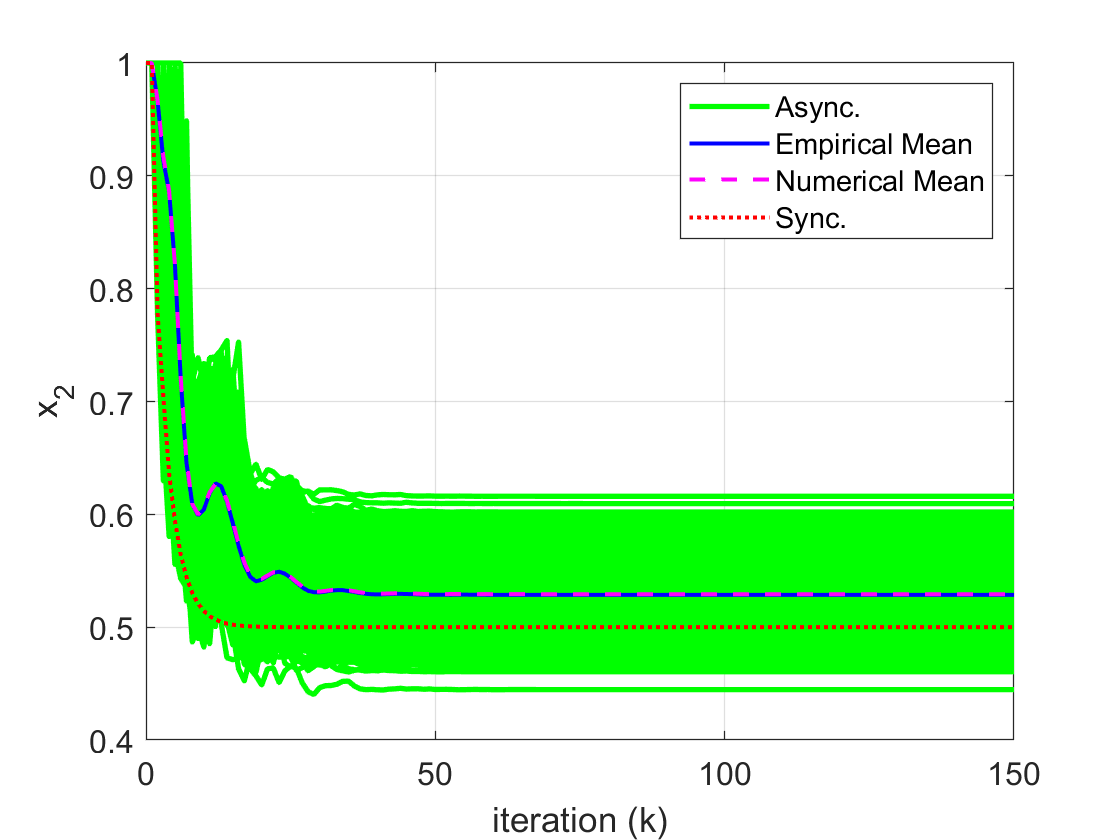}}\,
\subfloat[]{\includegraphics[scale=0.2]{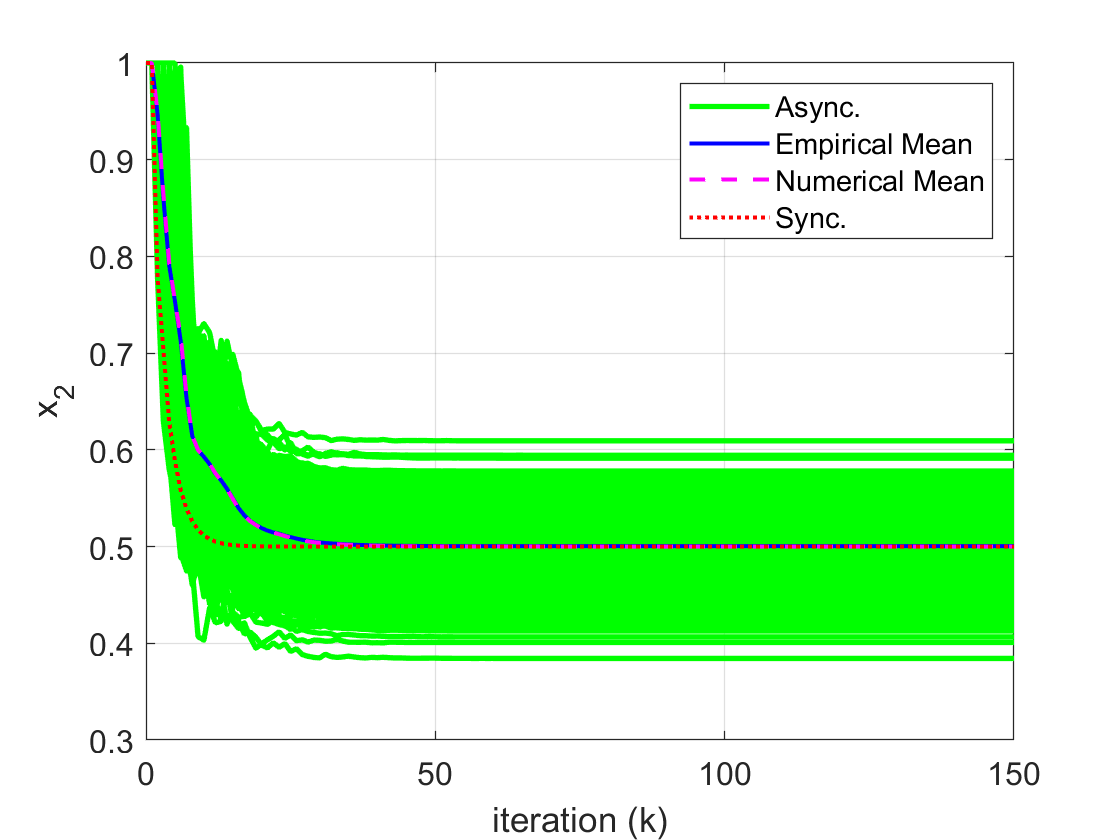}}
\caption{Simulation results for the asynchronous distributed averaging: (a) network topology for the six-node system; time evolution of $x_2$ values under asynchronous updates for (b) non-identical diagonal elements and (c) identical diagonal elements cases}\label{fig: sim}
\end{figure*}
In this section, simulation results are provided to test the technical soundness of the proposed results. 
The network topology for the simulation of the distributed averaging is depicted in Fig. \ref{fig: sim} (a) that is a six-node system only connected with adjacent nodes. Two different cases are considered for (b) non-identical  and (c) identical diagonal elements $\{a_{ii}\}$ in $A$. The associated matrix $A$ for each case is given below:

\vspace{-0.3in}
\small
\begin{align*}
&A_{(b)}= & &A_{(c)}=\\
&\begin{bmatrix}
\sfrac{1}{3} & \sfrac{1}{3} & 0 & 0 & 0 & \sfrac{1}{3}\\
\sfrac{1}{3} & \sfrac{1}{3} & \sfrac{1}{3} & 0 & 0 & 0\\
0 & \sfrac{1}{3} & \sfrac{1}{3} & \sfrac{1}{3} & 0 & 0\\
0 & 0 & \sfrac{1}{3} & \sfrac{5}{12} & \sfrac{1}{4} & 0\\
0 & 0 & 0 & \sfrac{1}{4} & \sfrac{1}{2} & \sfrac{1}{4}\\
\sfrac{1}{3} & 0 & 0 & 0 & \sfrac{1}{4} & \sfrac{5}{12}
\end{bmatrix},&\,
&\begin{bmatrix}
\sfrac{1}{3} & \sfrac{1}{3} & 0 & 0 & 0 & \sfrac{1}{3}\\
\sfrac{1}{3} & \sfrac{1}{3} & \sfrac{1}{3} & 0 & 0 & 0\\
0 & \sfrac{1}{3} & \sfrac{1}{3} & \sfrac{1}{3} & 0 & 0\\
0 & 0 & \sfrac{1}{3} & \sfrac{1}{3} & \sfrac{1}{3} & 0\\
0 & 0 & 0 & \sfrac{1}{3} & \sfrac{1}{3} & \sfrac{1}{3}\\
\sfrac{1}{3} & 0 & 0 & 0 & \sfrac{1}{3} & \sfrac{1}{3}
\end{bmatrix}.
\end{align*}
\normalsize
The initial node values are set up as $x(0) = [1,1,1,0,0,0]^{T}$ and hence, the exact average is known to be 0.5.

A total of $1,000$ simulations were carried out for each case, where the ensembles for the node 2 from Monte Carlo simulations are presented in Fig. \ref{fig: sim} (b), (c) by green solid lines. The i.i.d. probability $\pi$ is randomly generated initially to reflect the property of asynchronous updates in distributed averaging and is fixed throughout all simulations. Due to the asynchronism, the node value converged to different averages for each run, whereas the synchronous distributed averaging ended up with the exact average -- $0.5$. The empirical and numerical means, respectively, are denoted by blue solid and magenta dashed lines. Notice that the numerical mean is analytically obtained by the result in Proposition \ref{prop:Wtilde_reduced_model}, which does not require Monte Carlo simulations. 

In Fig. \ref{fig: sim} (b), corresponding to the non-identical $\{a_{ii}\}$ case, the actual expected average error, $\lvert \mathbb{E}\left[\bar{x} - x^{\star}\right]\rvert$, is 0.0323 while its upper bound is calculated from Theorem \ref{thm: 1} by 0.0723. This is sufficiently valuable information to estimate the range of the exact average when updated asynchronously. 
Also, it is observed in Fig. \ref{fig: sim} (c) that the expected average error is zero in case of identical diagonal elements, meaning the expected asynchronous average is the same as the exact average. This proves the correctness of Corollary \ref{cor: 1}.

\section{Conclusion}
In this study, the expected average error for asynchronous distributed averaging problems was investigated. The upper bound for this error was derived based on the switched system framework. As a special case, it is shown that the expected average error is zero when the diagonal elements in the interaction topology matrix $A$ are identical across all nodes. 
The proposed method does not require a node to know the entire information about the interaction topology for the upper bound analysis. Further, it is unnecessary to detect opinions of each node, and hence privacy-preserving. 
The validation is performed by numerical examples.
%

\bibliographystyle{ieeetr}        
\bibliography{reference}           



\end{document}